\pdfoutput=1
\documentclass[sigconf]{acmart}

\AtBeginDocument{%
  \providecommand\BibTeX{{%
    \normalfont B\kern-0.5em{\scshape i\kern-0.25em b}\kern-0.8em\TeX}}}

\setcopyright{none}

\usepackage{subcaption}

\usepackage[ruled]{algorithm2e}
\usepackage{pgfplots}
\usetikzlibrary{shapes,snakes}

\input{Definitions}

\begin{document}


\title[Tiering as Submodular Optimization]{Tiering as a Stochastic Submodular Optimization Problem}




\author{Hyokun Yun}
\affiliation{%
  \institution{Amazon.com}
  \streetaddress{550 Terry Ave N.}
  \city{Seattle}
  \state{Washington}
  \postcode{98109}
}
\email{yunhyoku@amazon.com}

\author{Michael Froh}
\affiliation{%
  \institution{Amazon.com}
  \streetaddress{550 Terry Ave N.}
  \city{Seattle}
  \state{Washington}
  \postcode{98109}
}
\email{froh@amazon.com}

\author{Roshan Makhijani}
\affiliation{%
  \institution{Amazon.com}
  \streetaddress{550 Terry Ave N.}
  \city{Seattle}
  \state{Washington}
  \postcode{98109}
}
\email{makhijan@amazon.com}

\author{Brian Luc}
\affiliation{%
  \institution{Amazon.com}
  \streetaddress{550 Terry Ave N.}
  \city{Seattle}
  \state{Washington}
  \postcode{98109}
}
\email{briluc@amazon.com}

\author{Alex Smola}
\affiliation{%
  \institution{Amazon.com}
  \streetaddress{550 Terry Ave N.}
  \city{Seattle}
  \state{Washington}
  \postcode{98109}
}
\email{smola@amazon.com}

\author{Trishul Chilimbi}
\affiliation{%
  \institution{Amazon.com}
  \streetaddress{550 Terry Ave N.}
  \city{Seattle}
  \state{Washington}
  \postcode{98109}
}
\email{trishulc@amazon.com}








\renewcommand{\shortauthors}{Yun et al.}

\begin{abstract}
  Tiering is an essential technique for building large-scale
  information retrieval systems.  While the selection of documents for
  high priority tiers critically impacts the efficiency of tiering,
  past work focuses on optimizing it with respect to a static set of
  queries in the history, and generalizes poorly to the future
  traffic. Instead, we formulate the optimal tiering as a stochastic
  optimization problem, and follow the methodology of regularized
  empirical risk minimization to maximize the \emph{generalization
    performance} of the system. We also show that the optimization
  problem can be cast as a stochastic submodular optimization problem
  with a submodular knapsack constraint, and we develop efficient
  optimization algorithms by leveraging this connection.
\end{abstract}






\maketitle

\section{Introduction}

Tiering \citep{risvik2003multi} is a classic method for scaling
information retrieval systems to large corpora by restricting searches
on a subset of documents.
In this paper, we focus on two-tier systems because of their popularity
and simplicity, but ideas in this paper can be applied to more than
two tiers by iteratively splitting a tier into two. 
In two-tier systems, Tier 1 indexes and serves only a selected set of
documents in the corpus, whereas Tier 2 covers all documents. For
every incoming query, a query classifier checks the \emph{eligibility}
of the query to Tier 1. If the query is eligible, it is handled by
Tier 1, which is more efficient than Tier 2 in serving the query
because Tier 1 indexes much fewer documents than Tier 2 does. However,
there are queries which documents in Tier 1 are not sufficient to
serve, and they are handled by Tier 2.

The selection of documents for Tier 1 critically impacts the
efficiency of tiering. The set of documents indexed in Tier 1 should
be small enough to make searches on the tier more efficient than
searches on Tier 2, while still being able to faithfully serve a good
fraction of traffic in order to justify the opportunity cost of adding
more compute power to Tier 2 instead.  While early approaches used
heuristic rules \citep{risvik2003multi, baeza2009efficiency} for the
selection, \citet{leung2010optimal} formulated the selection problem
as a max-flow problem, and developed an efficient optimization
algorithm for it.

However, these approaches optimize their selection with respect to a
static set of queries from the past. Due to the heavy-tailed nature of
query distributions in information retrieval, a large fraction of
queries in the incoming traffic are novel ones which never appeared in
the data from the past \citep{baeza2007impact}.  Therefore, these
approaches can overfit to the ``training data'' used for the
optimization, and generalize poorly to the incoming traffic.

In order to optimize the tiered architecture in terms of its
generalization performance on the future traffic, we first propose a
novel method of tiering which a tier is defined with the set of
clauses it comprehensively indexes.  Then, we formulate the task of
finding the best selection of clauses as a stochastic optimization
problem, and follow the methodology of regularized empirical risk
minimization to optimize the generalization performance. While this
involves a combinatorial optimization, we also show that the problem
can be cast as a stochastic submodular maximization problem with a
submodular knapsack constraint, which is called Submodular Cost
Submoudlar Knapsack (SCSK) \cite{iyer2013submodular}. We also develop
efficient optimization algorithms for the SCSK problem which make it
possible to apply the proposed tiering method on large-scale corpora
in the real world.



This paper makes contributions to both machine learning and
information retrieval community. For the machine learning research, we
introduce a novel application of large-scale stochastic submodular
optimization, orders of magnitudes larger than applications considered
in the literature. We also develop efficient algorithms for the SCSK
problem, which can be used in other applications and thus are of
independent interest.  For the information retrieval community, we
show the importance of formulating the tiering as a learning problem,
and develop an efficient tiering method. We demonstrate the practical
utility of our approach by applying our method to real data from a
large-scale commercial search engine.

\section{Formulation and Previous Work}

\subsection{Matching}

In this paper, we focus on the task of finding the set of documents
that match every term in a query. This is called \emph{matching}, and
matching algorithm is often the most computationally challenging
component in scaling a search engine to large corpora
\citep{goodwin2017bitfunnel}.
To be specific, let $V$ be the possibly infinite vocabulary which
contains all terms that appear in the corpus or a query. Let $D$ be
the corpus to search on. Each document $d \in D$ is represented as a
set of terms in it, i.e., $d \subseteq V$. Let $\Qcal$ be the
probability distribution of queries, and each query $q$ drawn from
$\Qcal$ is also represented as a subset of $V$. Then, the match set of
a query $q$ is defined as the set of documents which contain all terms
in the query:
\begin{align}
  m(q) := \cbr{ d \in D; q \subseteq d} = \bigcap_{v \in q} \cbr{d \in D; v \in d}.
  \label{eq:match}
\end{align}
The goal of a matching algorithm is to efficiently compute the match
set (\ref{eq:match}) over an intimidatingly large corpus. Then,
documents in the match sets are iteratively sorted and filtered by
more sophisticated relevance algorithms to determine the final
presentation to the user.  Since these post-processing algorithms are
too expensive to apply to the entire corpus, they rely on matching
algorithms for drastically reducing the scope of analysis. This is why
the efficiency of a matching algorithm is a critical determinant of
the scalability of the information retrieval system
\citep{risvik2013maguro, goodwin2017bitfunnel}.


\begin{table}
  \begin{center}
    {\small
    \begin{tabular}{c|ccccc}
      \hline \hline
      & \texttt{red} & \texttt{blue} & \texttt{shirt} & \texttt{pants} & \texttt{striped} \\
      \hline
      $D_1$ (red shirt striped) & X & & X & & X \\
      $D_2$ (blue shirt striped) &  & X & X & & X \\
      $D_3$ (red shirt) & X &  & X & &  \\
      $D_4$ (red pants striped) & X & &  & X & X \\
      $D_5$ (blue pants striped) &  & X &  & X & X \\
      $D_6$ (blue pants) &  & X &  & X  &  \\      
      \hline \hline     
    \end{tabular}
    }
  \end{center}
  \caption{An Example Corpus $\Dcal = \cbr{D_1, D_2, \ldots, D_6}$. X
    denotes that the document in the corresponding row contains the
    word in the corresponding column. }
  \label{tab:example_corpus}
\end{table}

To illusrate the concept of matching with examples, consider a corpus
with six documents, shown in Table~\ref{tab:example_corpus}. The match
set of a query ``red shirt'' would be the intersection of postings
list of the word \texttt{red} and that of \texttt{shirt}:
$m(\cbr{ \texttt{red}, \texttt{shirt} }) = \cbr{D_1, D_3, D_4} \cap
\cbr{D_1, D_2, D_3} = \cbr{D_1, D_3}$. Similarly,
$m(\cbr{ \texttt{blue}, \texttt{pants}, \texttt{striped} }) =
\cbr{D_2, D_5, D_6} \cap \cbr{D_4, D_5, D_6} \cap \cbr{D_1, D_2, D_4,
  D_5} = \cbr{D_5}$. Note that documents in the match set are not yet
ordered. The ordering of documents presented to a user is determined
by a ranking algorithm which take the match set as an input.

\subsection{Tiering}

Tiering \citep{risvik2003multi} improves the efficiency of finding
match sets of queries by selectively reducing the scope of search into
a subset of documents in the corpus.  In this paper, we focus on the
case of two tiers, which is standard and simple to illustrate.
Designing a tiered architecture requires defining two components: a
document classifier and a query classifier.

Let $\phi: \Dcal \rightarrow \cbr{1, 2}$ be the \emph{document
  classifier}, which determines the lowest tier an input document
should be indexed in. Tier 1 indexes
$\Dcal^1 := \cbr{ d \in \Dcal; \phi(d) \leq 1 }$ only, whereas Tier 2
indexes every document;
$\Dcal^2 := \cbr{ d \in \Dcal; \phi(d) \leq 2 } = \Dcal$.  Let
$\abr{\cdot}$ be the cardinality of a set.  In order for queries to be
more efficiently executed in Tier 1 than in Tier 2, $\abr{ \Dcal^1 }$
should be substantially smaller than $\abr{ \Dcal }$.  For example,
documents in a large corpus are often partitioned into shards
\citep{barroso2003web} such that each partition of the index assigned
to a shard fits into the memory capacity of a single machine. In such
case, half-sized Tier 1 index (i.e.,
$\abr{\Dcal^1} < \frac{\abr{\Dcal^2}}{2}$) would require half the
number of machines needed by Tier 2 in processing each query.

On the other hand, let $\psi: \Qcal \rightarrow \cbr{1, 2}$ be the
\emph{query classifier}, which determines the tier the query should be
executed in.  Note that Tier 1 is unable to return the entire match
set of a query $q$ unless $m(q) \subseteq \Dcal^1$; whenever a query
is classified into Tier 1 but incomplete match set was found, i.e.,
$\psi(q) = 1$ but $m(q) \nsubseteq \Dcal^1$, we consider this as an
\emph{incorrect} query classification.  Some tiered systems don't
guarantee their query classifications to be always correct
\citep{risvik2003multi, baeza2009efficiency}.

Since any document missed by a matching algorithm will also be missed
in the final search result presented to the user, it is critical for a
matching algorithm to have a nearly perfect recall
\citep{ntoulas2007pruning, goodwin2017bitfunnel}. Incorrect query
classifications to Tier 1, however, incur false negative errors. In
applications which the searchability of a document should be
guaranteed, these errors can be intolerable. Therefore, we focus on
developing methods which always make correct query classifications and
return the comprehensive match set for every query. See
Figure~\ref{fig:tiering} for a graphical illustration.


\begin{figure}
  \begin{tikzpicture}[block/.style={rectangle,draw}]
    \node [block] (query) {Query $q$};
    \node [ellipse, minimum height=1cm, draw] (tier1) [below of=query, xshift=-3cm, yshift=-1cm] {\begin{tabular}{c}Tier 1\\($\Dcal_1$)\end{tabular}};
    \node [ellipse, minimum height=2cm, draw] (tier2) [below of=query, xshift=3cm, yshift=-1cm] {\begin{tabular}{c}Tier 2\\($\Dcal_2 = \Dcal$)\end{tabular}};
    \path [draw, ->] (query) edge [left, align=right, bend right=10] node {$\psi(q) = 1$} (tier1);
    \path [draw, ->] (tier1) edge [right, align=right, bend right=10, pos=0.1] node {$m(q) \subset \Dcal_1$} (query);
    \path [draw, ->] (query) edge [left, align=right, bend right=10] node {$\psi(q) = 2$} (tier2);
    \path [draw, ->] (tier2) edge [right, align=right, bend right=10] node {$m(q) \subset \Dcal_2$} (query);
    \node [block, draw] (document) [below of=query, yshift=-3cm] {Document $d$};

    \path [draw, ->] (document) edge [left, align=right] node {$\phi(d) = 1$} (tier1);
    \path [draw, ->] (document) edge [right, align=left] node {$\phi(d) = 1$ or $2$} (tier2);
  \end{tikzpicture}
  \caption{Illustration of Tiering.  At the indexing time, all
    documents are indexed in Tier 2, and only documents with
    $\phi(d) = 1$ are \emph{additionally} indexed in Tier 1, which is
    thus smaller than Tier 2 and more effective in serving queries.
    On the other hand, an incoming query is routed to either Tier 1 or
    Tier 2, depending on its query classification result $\psi(q)$. }
  \label{fig:tiering}
\end{figure}
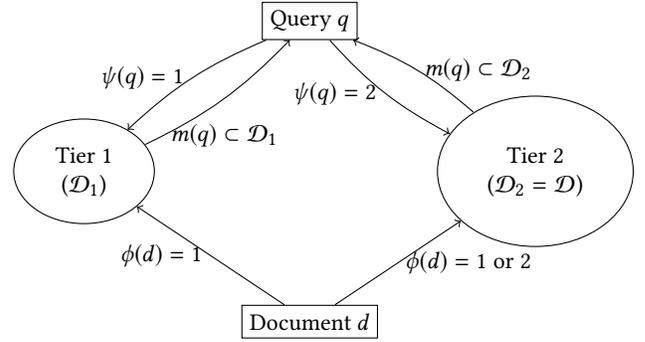

\subsection{Maximum Flow-Based Optimization}

In order to maximize the efficiency of a tiered architecture,
$\Dcal^1$, documents for Tier 1, should be carefully chosen so that a
good fraction of traffic can be efficiently handled by Tier 1.  We
formulate this as the following stochastic optimization problem, which
objective is to maximize the fraction of traffic covered by Tier 1:
\begin{align}
  &\max_{\psi, \phi}
    \PP_{q \sim \Qcal} \sbr{ \psi(q) = 1 },
    \text{ subject to } \label{obj} \\
  &\;\;\; \abr{\Dcal_1} = \abr{ \cbr{ d \in \Dcal;  \phi(d) = 1 } } \leq B, \label{capconst} \\
  &\;\;\; \text{for every query } q \text{ with } \psi(q) = 1, \phi(d) = 1 \text{ for every } d \in m(q). \label{corconst}
\end{align}
$B$ in constraint \eqref{capconst} denotes the capacity of Tier 1, and
the constraint \eqref{corconst} represents the correctness of query
classifications.

Assuming we have an access to sample queries of size $n$, which we
denote as $Q_n$, \citet{leung2010optimal} approaches the optimal
tiering as an optimization problem of selecting a subset $X$ of $Q_n$
which will be eligible for Tier 1\footnote{We simplified the objective
  proposed in \citet{leung2010optimal} for the case of two tiers.}:
\begin{align}
  \max_{X \subset Q_n} \frac{\abr{X}}{n}
  \text{ subject to }
  \abr{ \bigcup_{q \in X} m(q) } \leq B.
  \label{eq:flowobj}
\end{align}
Although solving the original problem \eqref{eq:flowobj} requires an
integer programming and thus it is practically infeasible even at
moderate scale, they develop efficient solvers by replacing the
constraint with a partial Lagrangian, and optimizing the convex
relaxation of the objective. Since this approach is equivalent to
solving a maximum flow problem over a directed graph, we call this
method \texttt{flow} whenever an abbreviation is needed.

With a solution of \eqref{eq:flowobj}, which we denote as
$X^{\texttt{flow}}$, the query classifier and the document classifier
of \texttt{flow} can be defined as follows:
\begin{align}
  \psi^{\texttt{flow}}(q)=
  \begin{cases}
    1, & \text{ if } q \in X^{\texttt{flow}} \\
    2, & \text{otherwise} \\
  \end{cases}
\end{align}
\begin{align}
  \phi^{\texttt{flow}}(d)=
  \begin{cases}
    1, & \text{ if } d \in m(q) \text{ for some } q \in X^{\texttt{flow}} \\
    2. & \text{otherwise} \\
  \end{cases}
\end{align}
It is easy to see that if $X^{\text{flow}}$ is a solution of
\eqref{eq:flowobj}, then induced query classifier
$\psi^{\texttt{flow}}(q)$ and document classifier
$\phi^{\texttt{flow}}(d)$ satisfy constraints (\ref{capconst}) and
(\ref{corconst}). However, the objective of (\ref{eq:flowobj}) is
\begin{align*}
  \frac{\abr{X^{\texttt{flow}}}}{n} = \frac{\abr{ \cbr{q \in Q_n ; \psi^{\texttt{flow}}(q) = 1} }}{n}
  = \PP_{q \sim \Qcal_n} \sbr{\psi(q)^{\texttt{flow}} = 1},
\end{align*}
which is the probability of a query covered by Tier 1 in the
\emph{empirical distribution} $\Qcal_n$ induced by $Q_n$, rather than
in the original distribution $\Qcal$ itself. Consequently, optimizing
(\ref{eq:flowobj}) is at the risk of overfitting to the training data
$Q_n$.

Indeed, there are two clear limitations of this approach when its
generalization performance to incoming traffic is concerned. First,
any
query which did not appear in the training data, i.e., $q \notin Q_n$,
will be routed to Tier 2 because by design, the method can select
queries in the training data only; $X^{\texttt{flow}} \subseteq Q_n$.
Due to the heavy-tailed nature of query distributions, no matter how
large the training data $Q_n$ is, a considerable fraction of new
samples from $\Qcal$ would not overlap with $Q_n$
\citep{baeza2007impact}, all of which are missed opportunities for
\texttt{flow}. On the other hand, very specific queries in the
training data which are very unlikely to reappear, for example a query
like ``book on submodular optimization with red cover'', are more
likely to be selected because their match sets are small and therefore
it is easy to meet the correctness constraint with them.
Therefore, a tiering method which allows more effective optimization
of generalization performance is called for.

\subsection{Other Related Work}


There are other methods for optimizing the size of index for Tier 1.
Index pruning methods \citep{carmel2001static} reduce the size of the
index by removing postings which minimally impact the quality of the
search results.  While these methods can reduce the size of the
backward index, they don't have a direct control for the size of the
forward index. On the other hand, we focus on controlling the size of
the forward index by constraining the number of documents added to the
index, without directly optimizing the size of the backward index.
This is because commercial search engines need to return rich metadata
associated with each document in the search result, and therefore the
size of the forward index dominates the size of the backward index.
Conceptually, however, these methods are complementary can be combined
together.  Most index pruning methods do not guarantee the correctness
of query classifications, but \citet{ntoulas2007pruning} is a notable
exception which provides similar guarantees as ours; while they
propose simple heuristics for finding the optimal set of terms for
Tier 1, their optimization problem can also be formulated and
approached in the way similar to the method we discuss in
Section~\ref{sec:clause}.

\citet{anagnostopoulos2015stochastic} is also a relevant work which
formulates caching as a stochastic optimization problem. While we
focused on the generalization of query classifiers to the future
traffic, their work was concerned about stochastically modeling the
size of the cache needed to run their algorithm, and used the same
query classifier as in \citet{leung2010optimal} which cannot
generalize to queries unseen in the past.

\section{Tiering with Clause Selection}
\label{sec:clause}

In order to build a tiered architecture which can generalize to
queries unseen in the training data, we propose a novel method of
tiering which query and document classification decisions are made
with clauses, which are sets of terms and can be more granular than
full queries. This makes it possible for queries unseen in the
training data to be classified into Tier 1, as long as some clauses in
the query were observed.

\subsection{Query and Document Classifier}
\label{ssec:clause_classifier}

We parameterize query and document classifiers with
$X^{\text{clause}} \subseteq 2^V$, which is a subset of all possible
clauses in the vocabulary.  Then, the query classifier and the
document classifier identically check whether any of the clauses in a
query or a document is included in the selection:
\begin{align}
  \psi^{\texttt{clause}}(q) =
  \begin{cases}
    1, & \text{ if } c \subseteq q \text{ for some } c \in X^{\texttt{clause}} \\
    2, & \text{otherwise} \\
  \end{cases}
\end{align}
\begin{align}
  \phi^{\texttt{clause}}(d) =
  \begin{cases}
    1, & \text{ if } c \subseteq d \text{ for some } c \in X^{\texttt{clause}} \\
    2. & \text{otherwise} \\
  \end{cases}
\end{align}
Efficient algorithms for computing these ``subset query'' functions
exist, e.g., \citet{charikar2002new} or \citet{savnik2013index}, which
make it possible for these classifiers to be used in applications with
low latency requirements.

To illustrate these classifiers, consider the example corpus in
Table~\ref{tab:example_corpus} and suppose we chose two clauses
$X^{\texttt{clause}} = \cbr{ \cbr{\texttt{red}}, \cbr{\texttt{blue},
    \texttt{shirt}} }$.  Then, documents which contain a single-word
clause ``red'' or a double-word clause ``blue shirt'' will be
classified into Tier 1: $\Dcal^1 = \cbr{D_1, D_2, D_3, D_4}$. With
these documents, Tier 1 shall serve queries such as ``red'', ``red
shirt'', ``red pants'', or ``blue shirt striped'', but not ``blue
pants'', because neither $\cbr{\texttt{red}}$ nor
$\cbr{\texttt{blue}, \texttt{shirt}}$ is a subset of
$\cbr{\texttt{blue}, \texttt{pants}}$. Indeed,
$D_6 (\text{blue pants})$ is not included in $\Dcal_1$.

Clearly, the query classifier is always correct:
\begin{theorem}[Correctness]
  \label{thm:correct}
  For any $q \subseteq V$ with $\psi^{\texttt{clause}}(q) = 1$, we
  have $\phi^{\texttt{clause}}(d) = 1$ for any $d \in m(q)$.
  This implies $m(q) \subseteq \Dcal_1$.
\end{theorem}
\begin{proof}
  Suppose $\psi^{\texttt{clause}}(q) = 1$ for some $q \subseteq V$.
  Then, from the definition of $\psi^{\texttt{clause}}(q)$, there
  exists some $c \in X^{\texttt{clause}}$ which satisfies
  $c \subseteq q$. Now pick any $d \in m(q)$. By definition
  \eqref{eq:match}, we have $q \subseteq d$. Therefore,
  $c \subseteq d$ and $\phi^{\texttt{clause}}(d) = 1$.
\end{proof}

\subsection{Stochastic Submodular Optimization}

In order to maximize the efficiency of the tiered architecture, now we
need to solve (\ref{obj}) parameterized by
$X^{\texttt{clause}}$. Since Theorem~\ref{thm:correct} guarantees the
correctness, the corresponding constraint can be dropped and the
problem can be simplified as follows:
\begin{align}
  &\max_{X^{\texttt{clause}} \subseteq 2^V}
    \PP_{q \sim \Qcal}\sbr{ c \subseteq q \text{ for some } c \in X^{\texttt{clause}} },
  \label{eq:submod_obj} \\
  &\;\;\;\;\;\text{ subject to }
    \abr{ \cbr{d \in \Dcal; c \subseteq d \text{ for some } c \in X^{\texttt{clause}}}  } \leq B.
    \label{eq:submod_const}
\end{align}
While this may seem as an intractable combinatorial problem as every
subset of the power set of $V$ should be considered as a potential
solution, we argue that there is a strong structure in the problem we
can exploit for developing efficient optimization
algorithms. Specifically, we show that this is a stochastic version
\citep{karimi2017stochastic} of the Submodular Cost Submodular
Knapsack (SCSK) problem \citep{iyer2013submodular}, which aims to
maximize a stochastic submodular function under an upper bound
constraint on another submodular function.  This observation makes it
promising to develop practical algorithms for the problem, as
optimization problems formulated with submodular functions often allow
practical algorithms with theoretical guarantees.

To prove that (\ref{eq:submod_obj}) is a SCSK problem, we first remind
readers basic notions of submodularity
\citep{fujishige2005submodular}:
\begin{definition}[Gain, Monotonocity, and Submodularity]
  Let $G$ be a set, and $f: 2^G \rightarrow \RR$ be a set
  function. The \emph{gain} of $j \in G$ at $Y$ is defined as
  $f(j \mid Y) := f(Y \cup \cbr{j}) - f(Y)$. If $f(j \mid Y) \geq 0$
  for any $j \notin Y$, then $f(\cdot)$ is called \emph{monotone}.
  Submodular functions have diminishing gains; that is, $f(\cdot)$ is
  called \emph{submodular} if $f(j \mid Y) \geq f(j \mid Z)$ for any
  $Y$, $Z$ with $Y \subseteq Z$ and $j \notin Y$.
\end{definition}

Now we prove that the objective function is submodular.
\begin{theorem}[Submodularity of the Objective]
  The objective function (\ref{eq:submod_obj}),
  $f(X) := \PP_{q \sim \Qcal} \sbr{ c \subseteq q \text{ for some } c
    \in X }$, is a monotone submodular function of $X$.
\end{theorem}
\begin{proof}
  First, pick any $q \subseteq V$, and define
  \begin{align*}
    f_q(X) := 1\cbr{c \subseteq q \text{ for some } c \in X},
  \end{align*}
  where $1\cbr{\cdot}$ is an indicator function. We start by showing
  monotone submodularity of this function.

  The monotonicity of the function can be shown by noting that:
  \begin{align*}
    f_q(j \mid Y) = 
    \begin{cases}
      0, & \text{ if } c \subseteq q \text{ for some } c \in Y \\
      0, & \text{ if } c \nsubseteq q \text{ for any } c \in Y \cup \cbr{j} \\
      1, & \text{ if } c \nsubseteq q \text{ for any } c \in Y \text{ but } j \subseteq q\\
    \end{cases}    
  \end{align*}
  and thus $f_q(j \mid Y) \geq 0$ for every case.

  Now suppose $Y, Z$ are given, with $Y \subseteq Z \subseteq 2^V$. We
  have three cases.
  \begin{itemize}
  \item Case 1: There exists some $c \in Y$ which $c \subseteq q$.
    Then, $f_q(Y) = f_q(Z) = 1$, and $f_q(j \mid Y) = f_q(j \mid Z) = 0$ as
    indicator functions are bounded above by 1.
  \item Case 2: There
    does not exist any $c \in Z$ which $c \subseteq q$.
    In this case, $f_q(Y) = f_q(Z) = 0$, and 
    $f_q(j \mid Y) = f_q(j \mid Z) = 1$ if $j \subseteq q$, and
    $f_q(j \mid Y) = f_q(j \mid Z) = 0$ otherwise.
  \item Case 3: There does not exist any $c \in Y$ with
    $c \subseteq q$, but there exists $c' \in Z$ which
    $c' \subseteq q$. Then, $f_q(Y) = 0$ and $f_q(Z) = 1$.  Since
    $f_q(\cdot)$ is bounded above by 1, $f_q(j \mid Z) = 0$.  Because
    $f_q(\cdot)$ is monotonic, $f_q(j \mid Y) \geq f_q(j \mid Z)$.
  \end{itemize}
  Therefore, $f_q(\cdot)$ is monotone submodular.

  In order to prove the monotone submodularity of $f(\cdot)$, note
  that $f(X) = \EE_{q \sim \Qcal} f_q(X)$. A convex combination of
  monotone submodular functions is again monotone submodular.
\end{proof}

\begin{theorem}[Submodularity of the Constraint]
  $g(X) := \abr{ \cbr{d \in \Dcal; c \subseteq d \text{ for some } c \in X}  }$ is monotone submodular.
\end{theorem}
\begin{proof}
  Observe that
  \begin{align*}
    \cbr{d \in \Dcal; c \subseteq d \text{ for some } c \in X} =
    \bigcup_{c \in X} \cbr{d \in \Dcal; c \subseteq d}
    = \bigcup_{c \in X} m(c).
  \end{align*}
  Therefore, $g(X)$ is a set covering function, which is known to be
  monotone submodular \citep{wolsey1982analysis}.
\end{proof}

\subsection{Regularized Empirical Risk Minimization}

Since we typically don't have an access to the true query distribution
$\Qcal$, we need to use our training data $Q_n$ and optimize
(\ref{eq:submod_obj}) with respect to the induced empirical
distribution $\Qcal_n$ instead of $\Qcal$. In order to avoid
overfitting to the training data, we follow the methodology of
regularized risk minimization, and control the capacity of the
function we learn \cite{vapnik1992principles}. Specifically, we
restrict the ground set of the objective function to be
$\bar{X} := \cbr{c \in 2^V; \PP_{q \sim \Qcal_n}\sbr{c \subseteq q}
  \geq \lambda}$, so that only clauses which frequency of appearance
in the training data is at least $\lambda$ are considered. The
optimization problem then becomes
\begin{align}
  &\max_{X \subseteq \bar{X}}
    f(X) := \PP_{q \sim \Qcal_n}\sbr{ c \subseteq q \text{ for some } c \in X },
    \label{eq:emp_obj}\\
  &\;\;\;\;\;\text{ subject to }
    g(X) := \abr{ \bigcup_{c \in X} m(c)}  \leq B \nonumber .
\end{align}
Not only does this regularize the solution for improved generalization,
it also has a computational benefit 
because the number of clauses to be considered is drastically reduced.
The time complexity of most submodular optimization algorithms are at
least linear to the cardinality of the ground set.  Also note that
$\bar{X}$ can be efficiently computed from $Q_n$ using frequent
pattern mining algorithms, and we use \texttt{FPGrowth}
\citep{han2000mining} in our experiments.

\section{Optimization Algorithms for SCSK}
\label{sec:scsk}

Solving \eqref{eq:emp_obj} for large-scale information retrieval
systems is a computational challenge, since the number of documents to
be indexed $\abr{\Dcal}$ and the number of query logs available
$\abr{Q_n}$ in such systems are often at formidable scale
($10^6 \sim 10^{12}$) \citep{risvik2013maguro}. In order to achieve a
high coverage of traffic with Tier 1, we also consider a large number
of clauses $\abr{\bar{X}}$ at the scale of $10^4 \sim 10^6$, orders of
magnitude higher than problems considered to be large-scale in
submodular optimization research (e.g., \citep{iyer2013submodular,
  karimi2017stochastic}). Therefore, in order to apply the proposed
tiering method to real-world problems, it is critical to develop
efficient and scalable optimization algorithms for \eqref{eq:emp_obj}.
In this section, we develop multiple algorithms for the problem which
can be broadly applicable to other SCSK problems, and therefore of
interest on their own.

\subsection{Greedy}

Greedy algorithms are often very competitive at solving a wide range
of submodular maximization problems (e.g.,
\citep{nemhauser1978analysis, sviridenko2004note, wei2015mixed,
  bai2016algorithms}). Because of the submodular upper bound
constraint in the problem (\ref{eq:emp_obj}), however, most of
existing greedy algorithms for submodular maximization problems do not
directly apply. A notable exception is the greedy algorithm from
\citet{iyer2013submodular}, but it ignores the constraint when
comparing candidates to add to the solution; while this simplicity
facilitates the mathematical analysis of the algorithm, it is clearly
important to take the \emph{cost} of each clause into consideration,
and we empirically validate this in Section~\ref{ssec:exp_opt}.

To this end, we propose a novel greedy algorithm for SCSK, which
starts with an empty solution $X^0 = \emptyset$, and iteratively adds
a new clause with the highest utility ratio:

\begin{align}
  X^{t+1} \leftarrow X^{t} \cup \cbr{ j^{(t)} := \argmax_{j \in \bar{X}, g(X^t \cup \cbr{j}) \leq B} \frac{f(j \mid X^t)}{g(j \mid X^t)} }.
  \label{eq:greedy_update}
\end{align}
Computing $f(j \mid X^t)$ and $g(j \mid X^t)$ for every
$j \in \bar{X}$ at every iteration, however, is clearly not feasible
at the scale of problems we consider. This is because computing
$g(j \mid X^t)$ involve calculating intersections between $m(j)$ and
$\bigcup_{c \in X} m(c)$, which can be both large sets with millions
of elements, and computing $f(j \mid X^t)$ is also expensive for the
same reason with large scale training data.

\subsubsection{Lazy Greedy}
\label{ssec:lazygreedy}

The lazy evaluation technique has been essential to the success of
many large-scale submodular maximization algorithms (e.g.,
\citep{minoux1978accelerated, leskovec2007cost, ahmed2012fair}), as
they can effectively avoid the costly evaluation of gain on less
promising candidates.  The submodularity of the constraint $g(\cdot)$,
however, requires us to be more careful in applying the technique.

To illustrate, consider a simple case which $g(\cdot)$ is a modular
function; that is, there exists $w_0 \in \RR$ and
$\cbr{w_c \in \RR}_{c \in \bar{X}}$ such that
$g(X) = w_0 + \sum_{j \in X} w_j$. In this case, the problem reduces
to submodular knapsack, and the greedy procedure
(\ref{eq:greedy_update}) acquires strong guarantees
\citep{sviridenko2004note}.  The utility ratio
$\frac{f(j \mid X^t)}{g(j \mid X^t)}$ is nonincreasing in $t$, because
$f(\cdot)$ is sumbdoular and $g(j \mid X^t) = w_j$ is a constant.  The
classic lazy greedy algorithm \citep{minoux1978accelerated} exploits
this by maintaining a max heap of candidates sorted by the most recent
evaluation of the utility ratio, and re-evaluate the ratio for
candidates which are promising enough to be placed at the top of the
heap.  Since $g(\cdot)$ is also \emph{submodular} in SCSK, however,
the ratio can also increase over iterations, and thus this technique
is not directly applicable.


In order to leverage lazy evaluations for SCSK, we maintain the lower
bound of $g(j \mid X^t)$ as $\underline{g}(j \mid X^t)$ with the
following update rule:
\begin{align}
  \underline{g}(j \mid X^{t+1}) \leftarrow
  \max\rbr{0, 
  \underline{g}(j \mid X^{t}) - g(j^{(t)} \mid X^{t})}.
  \label{eq:lbd_update}
\end{align}
This allows us to efficiently update the lower bound when adding a new
clause to the solution. Only when the optimistic estimate of the
utility of a clause is good enough for consideration, we re-compute
the function to make the bound tight. We prove the correctness of this
update:
\begin{theorem}[Correctness of Updated Lower Bound]
  Suppose $g(\cdot)$ is a monotone submodular function, and
  $g(j \mid X^t) \geq \underline{g}(j \mid X^t)$. With the update rule
  \eqref{eq:lbd_update}, we have
  $g(j \mid X^{t+1}) \geq \underline{g}(j \mid X^{t+1})$.
\end{theorem}
\begin{proof}
  \begin{align*}
    &g(j \mid X^{t+1}) = g(X^{t+1} \cup \cbr{j}) - g(X^{t+1})\\
    &= \cbr{g(X^t) + g(j \mid X^t) + g(j^{(t)} \mid X^t \cup \cbr{j}} - \cbr{g(X^{t}) + g(j^{(t)} \mid X^{(t)}) } \\
    &= g(j \mid X^t) - g(j^{(t)} \mid X^{(t)}) + g(j^t \mid X^t \cup \cbr{j}) \\
    &\geq g(j \mid X^t) - g(j^{(t)} \mid X^{(t)})
      \geq \underline{g}(j \mid X^t) - g(j^{(t)} \mid X^{(t)}),
  \end{align*}
  using the definition of gain, the monotonicity of $g(\cdot)$ for the
  first inequality, and the assumption for the second
  inequality. Combining this with the fact that
  $g(j \mid X^{t+1}) \geq 0$ due to monotonicity of $g(\cdot)$, the
  proof is completed.
\end{proof}

Algorithm~\ref{alg:lazygreedy} shows the pseudo-code of the lazy
greedy algorithm we propose. $\overline{f}(j \mid X^t)$ is potentially
outdated as it equals to $f(j \mid X^s)$ for some $0 \leq s \leq t$,
and serves as an upper bound due to submodularity. We then maintain a
max heap of every feasible candidate, exactly compute the utility
ratio of the best candidate at the moment, and add it to the solution
if it is better than the second candidate in the heap. This allows us
to avoid computing the utility ratio of candidates which even
optimistic estimate of the ratio is lower than $j^{(t)}$, the clause
selected in the $t$-th iteration.

\SetKwProg{Fn}{Function}{}{end} 
\SetKw{Continue}{continue}

\SetKwData{maxheap}{MaxHeap}
\SetKwFunction{LazyGreedy}{LazyGreedy}%
\begin{algorithm}
  \Fn(){\LazyGreedy{$f$, $g$, $B$, $\bar{X}$}}{
    $t=0$, $X^0 = \emptyset$\; $\overline{f}(j \mid X^0) = f(\cbr{j})$, $\underline{g}(j \mid X^0) = g(\cbr{j})$ for every $j \in \bar{X}$\;
    \While{$g(X^t) \leq B$}{
      \texttt{heap} $\leftarrow$ Max heap with items $\cbr{j \in \bar{X}; \underline{g}(X^t \cup \cbr{j}) \leq B}$,  and score $\frac{\overline{f}(j \mid X^t)}{\underline{g}(j \mid X^t)}$ \;
      \While{\texttt{heap.size() > 0}}{
        $j \leftarrow$ \texttt{heap.pop()}\;
        \tcp{Tighten bounds}
        $\overline{f}(j \mid X^t) \leftarrow f(j \mid X^t)$\;
        $\underline{g}(j \mid X^t) \leftarrow g(j \mid X^t)$\;
        \If{$g(X^t) + g(j \mid X^t) > B$} {
          \Continue
        }
        \eIf{\texttt{heap.isEmpty()} or
          $\frac{\overline{f}(j \mid X^t)}{\underline{g}(j \mid X^t)}
          \geq \frac{\overline{f}(k \mid X^t)}{\underline{g}(k \mid X^t)}$
          with $k \leftarrow$ \texttt{heap.peek()}} {
          $X^{t+1} \leftarrow X^t \cup \cbr{j}$\;
          $\overline{f}(i \mid X^{t+1}) = \overline{f}(i \mid X^t)$, $\underline{g}(i \mid X^{t+1}) = \max(0, \underline{g}(i \mid X^{t}) - g(j \mid X^{t}))$ for every $i \in \bar{X}$\;
          $t \leftarrow t + 1$\;
        }{
          \texttt{heap.push($j$)}\;
        }
      }
    }
  }
  \caption{Lazy Greedy Algorithm for SCSK}
  \label{alg:lazygreedy}
\end{algorithm}

\subsubsection{Optimistic-Pessimistic Parallel Greedy}
\label{ssec:optpesgreedy}

While the lazy evaluation procedure of Algorithm~(\ref{eq:lbd_update})
reduces the number of evaluations of $f(\cdot \mid X^t)$ and
$g(\cdot \mid X^t)$, the procedure is inherently sequential. In order
to leverage the compute power for parallel processing available in
modern computers, we propose an extension of the lazy evaluation
algorithm.  In addition to optimistic estimates
$\frac{\overline{f}(j \mid X^t)}{\underline{g}(j \mid X^t)}$ in the
lazy greedy algorithm, we also maintain \emph{pessimistic} estimates
$\frac{\underline{f}(j \mid X^t)}{\overline{g}(j \mid X^t)}$, where
$\underline{f}(j \mid X^t)$ is a lower bound of $f(j \mid X^t)$ using
the same update rule (\ref{eq:lbd_update}), and
$\overline{g}(j \mid X^t)$ is an upper bound of $g(j \mid X^t)$ using
outdated value of the gain. Then, we update the estimate of every
candidate which optimistic estimate is better than the best
pessimistic estimate. Algorithm~\ref{alg:optpes_greedy} illustrates
the pseudo-code.

\SetKwFunction{OptPesGreedy}{OptimisticPessimisticGreedy}%
\begin{algorithm}
  \Fn(){\OptPesGreedy{$f$, $g$, $B$, $\bar{X}$}}{
    $t=0$, $X^0 = \emptyset$\;
    $\overline{f}(j \mid X^0) \leftarrow \underline{f}(j \mid X^0) \leftarrow f(\cbr{j})$, $\overline{g}(j \mid X^0) \leftarrow \underline{g}(j \mid X^0) \leftarrow g(\cbr{j})$ for every $j \in \bar{X}$\;
    \While{$g(X^t) \leq B$}{
      $C \leftarrow \cbr{j \in \bar{X}; \underline{g}(X^t \cup \cbr{j}) \leq B, \frac{\overline{f}(j \mid X^t)}{\underline{g}(j \mid X^t)} \geq \max_{j'} \frac{\underline{f}(j' \mid X^t)}{\overline{g}(j' \mid X^t)}}$\;
      \tcp{Tighten bounds}
      \For{$j \in C$ in parallel}{
        $\overline{f}(j \mid X^t) \leftarrow \underline{f}(j \mid X^t) \leftarrow f(j \mid X^t)$\;
        $\overline{g}(j \mid X^t) \leftarrow \underline{g}(j \mid X^t) \leftarrow g(j \mid X^t)$\;
      }
      $j^{(t)} \leftarrow \arg\max_{j \in C, g(X^t \cup \cbr{j}) \leq B} \frac{f(j \mid X^t)}{g(j \mid X^t)}$\;
      $X^{t+1} \leftarrow X^t \cup \cbr{j^{(t)}}$\;
      $\overline{f}(i \mid X^{t+1}) = \overline{f}(i \mid X^t)$, $\overline{g}(i \mid X^{t+1}) = \overline{g}(i \mid X^t)$, $\underline{g}(i \mid X^{t+1}) = \max(0, \underline{g}(i \mid X^{t}) - g(j \mid X^{t}))$, $\underline{f}(i \mid X^{t+1}) = \max(0, \underline{f}(i \mid X^{t}) - f(j \mid X^{t}))$ for every $i \in C$\;
      $t \leftarrow t+1$\;
    }
  }
  \caption{Optimistic-Pessimistic Greedy Algorithm}
  \label{alg:optpes_greedy}
\end{algorithm}

We prove that this algorithm is consistent with the original greedy
update (\ref{eq:greedy_update}):
\begin{theorem}
  If $C$ is defined as in Algorithm~\ref{alg:optpes_greedy}, then
  $j^{(t)} \in C$.
\end{theorem}
\begin{proof}
  It is sufficient to prove
  $\frac{\overline{f}(j^{(t)} \mid X^t)}{\underline{g}(j^{(t)} \mid
    X^t)} \geq \max_{j'} \frac{\underline{f}(j' \mid
    X^t)}{\overline{g}(j' \mid X^t)}$. 
  Because of the definition of
  $j^{(t)}$ in (\ref{eq:greedy_update}),
  \begin{align*}
    \frac{\overline{f}(j^{(t)} \mid X^t)}{\underline{g}(j^{(t)} \mid X^t)}
    \geq
    \frac{f(j^{(t)} \mid X^t)}{g(j^{(t)} \mid X^t)}
    \geq
    \frac{f(j') \mid X^t)}{g(j') \mid X^t)}
    \geq
    \frac{\underline{f}(j') \mid X^t)}{\overline{g}(j') \mid X^t)},
  \end{align*}
  for any $j'$.
\end{proof}

\subsection{Iterative Submodular Knapsack}
\label{ssec:isk}

Iterative Submodular Knapsack is proposed by
\citet{iyer2013submodular}.  Inspired by the minorization-maximization
procedure from \citet{iyer2013fast}, $g(\cdot)$ is iteratively
approximated by a modular upper bound $\tilde{g}^{t+1}(\cdot)$, which
is exact at the current solution $X_t$; i.e.,
$g(X_t) = \tilde{g}^{t+1}(X_t)$, and $g(X) \leq \tilde{g}^{t}(X)$ for
every $X$ and $t$. They suggest two choices for the bound:
\begin{align}
  \tilde{g}_1^{t+1}(X) &:= g(X_t) - \sum_{j \in X_t \setminus X} g(j \mid X_t \setminus j) + \sum_{j \in X \setminus X_t} g(\cbr{j}), \nonumber \\
  \tilde{g}_2^{t+1}(X) &:= g(X_t) - \sum_{j \in X_t \setminus X} g(j \mid \bar{X} \setminus j) + \sum_{j \in X \setminus X_t} g(j \mid X_t).
  \label{eq:supergrad}
\end{align}
Given the choice of the modular upper bound, the solution is updated
as
\begin{align}
  X^{t+1} \leftarrow \arg\max_{X \in \bar{X}} f(X) \text{ subject to }
  \tilde{g}^{t}(X) \leq B,
  \label{eq:inner_knapsack}
\end{align}
which can be efficiently solved as a submodular knapsack problem
\cite{sviridenko2004note}.  Starting from $X^0 = \emptyset$, this
procedure is repeated until there is no change in the solution.
Algorithm~\ref{alg:isk} illustrates the pseudo-code.

\SetKwFunction{ISK}{ISK}%
\begin{algorithm}
  \Fn(){\ISK{$f$, $g$, $B$, $\bar{X}$, $m$}}{
    $t=0$, $X^0 = \emptyset$\;
    \While{$X^t \neq X^{t-1}$}{
      $X^{t+1} \leftarrow \arg\max_{X \in \bar{X}} f(X) \text{ subject to }
  \tilde{g}_m^{t}(X) \leq B$\;
    }
  }
  \caption{Iterative Submodular Knapsack Algorithm}
  \label{alg:isk}
\end{algorithm}

\section{Experiments}

We empirically evaluate the effectiveness of methods we propose in
this paper on real-world data. We use a corpus of about 8 million
documents in a particular category of a commercial search engine for
$\Dcal$, and we collected about 2 million queries uniformly sampled
over three days for the training data $Q_n$. We additionally sampled
about 700,000 queries in the following day for the test data.

We implemented every algorithm in Java, using its standard library for
multi-threading. We leveraged multi-threading in most of
straightforward opportunities for parallelization, most importantly
the computation of marginal gains and losses. In order to efficiently
compute set operations, we leveraged utilities provided by Apache
Lucene\footnote{\href{https://lucene.apache.org/}{https://lucene.apache.org/}}.
We also adopted techniques suggested in \citet{iyer2019memoization}
for incrementally computing marginal gains and losses of submodular
functions. All experiments were ran on a machine with 16 Intel Xeon
2.50GHz CPUs and 64GBs of RAM.

\subsection{Submodular Optimization Algorithms}
\label{ssec:exp_opt}

In this experiment, we focus on evaluating the performance of
submodular optimization algorithms we proposed in
Section~\ref{sec:scsk} for the SCSK problem. We set
$B = \frac{\abr{\Dcal}}{2}$, so that Tier 1 shall choose up to 50\%
documents of the whole corpus.  We consider following algorithms:

\begin{description}

\item[Constraint-Agnostic Greedy:] Greedy algorithm proposed in
  \citet{iyer2013submodular} for solving SCSK, which ignores the
  constraint function when comparing candidates. We implemented a lazy
  version of the algorithm \cite{minoux1978accelerated}.
  
\item[Greedy:] Greedily selects clauses according to the procedure
  \eqref{eq:greedy_update} we propose. Gains $f(j \mid X^t)$ and
  $g(j \mid X^t)$ are re-computed at every iteration.

\item[Lazy Greedy:] Algorithm~\ref{alg:lazygreedy}, which uses the
  same greedy procedure \eqref{eq:greedy_update} but improves
  efficiency by lazily evaluating candidates with max heap and
  optimistic estimates.
  
\item[Opt./Pes. Greedy:] Algorithm~\ref{alg:optpes_greedy}, which
  recomputes gains of clauses which optimistic estimate is better than
  the best pessimistic estimate.
  
\item[ISK:] Algorithm~\ref{alg:isk} from \ref{ssec:isk}. We use two
  variants ISK$_1$ and ISK$_2$, using each of the modular upper
  bound in \eqref{eq:supergrad}, respectively.
\end{description}

\begin{figure}
  \begin{tikzpicture}
    \begin{axis}[
      xlabel={Wall Clock Time (seconds)},
      ylabel={Tier 1 Query Coverage ($f(X)$)},
      yticklabels={,,},
      legend style={at={(0.5,1.1)},anchor=south},
      scaled y ticks = false]

      \addplot+[color=magenta, densely dotted, no markers, thick] table [col sep=comma, x=time,y=obj] {aggreedy.txt};
      \addlegendentry{Constraint-Agnostic Greedy}
      
      \addplot+[color=blue, no markers, thick] table [col sep=comma, x=time,y=obj] {naivegreedy.txt};
      \addlegendentry{Greedy}

      \addplot+[color=red, dashed, no markers, thick] table [col sep=comma,
      x=time,y=obj] {lazygreedy.txt}; \addlegendentry{Lazy Greedy}

      \addplot+[color=brown, densely dashed, no markers, thick] table [col sep=comma,
      x=time,y=obj] {optpesgreedy.txt}; \addlegendentry{Opt./Pes. Greedy}

      \addplot+[color=green, loosely dashed, no markers, thick] table [col sep=comma,
      x=time,y=obj] {isk2.txt}; \addlegendentry{ISK$_1$}

      \addplot+[color=black, dotted, thick, no markers, thick] table [col sep=comma,
      x=time,y=obj] {isk1.txt}; \addlegendentry{ISK$_2$}

    \end{axis}
  \end{tikzpicture}
  \caption{Comparison of the performance of optimization algorithms in
    terms of the objective function $f(X)$ as a function of wall clock
    time elapsed. Greedy and Lazy Greedy were terminated before
    reaching the stopping condition.  Units in the $y$-axis are hidden
    for the confidentiality of business.}
  \label{fig:wallclock}
\end{figure}
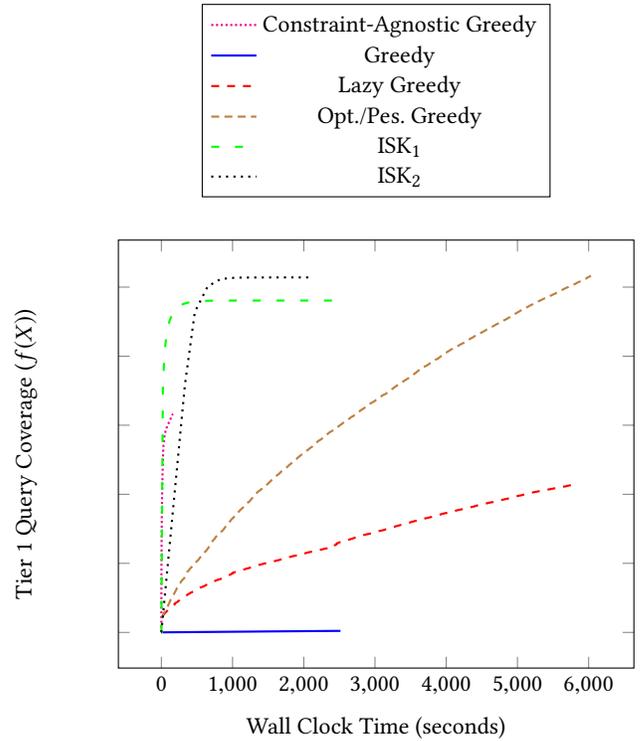

Figure~\ref{fig:wallclock} shows the value of the objective function
$f(X)$ with respect to the elapsed wall clock time.  It is noticeable
that ISK algorithms are much faster than most of greedy
algorithms. This is because first few iterations of ISK are much more
effective than those of greedy ones; the very first iteration of ISK
\eqref{eq:inner_knapsack} adds 28\% documents, whereas the greedy
algorithm needs to iterate hundreds of thousands of times to add the
same number of documents. While each iteration of ISK involves solving
a new submodular knapsack problem from scratch, the lazy evaluation
technique \citep{minoux1978accelerated} seems to be more efficient
with these submodular knapsack sub-problems than with the original
SCSK because only one function ($f(\cdot)$) is being approximated in
submodular knapsack, whereas we simultaneously bound two functions
$f(\cdot)$ and $g(\cdot)$ in SCSK.  On the other hand, the final
objective function value of the greedy algorithm was 7.6\% and 0.6\%
higher than that of ISK$_1$ and ISK$_2$, respectively. Therefore, the
greedy algorithm seems to be more effective at refining the
high-quality solution than ISK is, as each iteration of ISK relies on
a rough approximation of $g(\cdot)$.

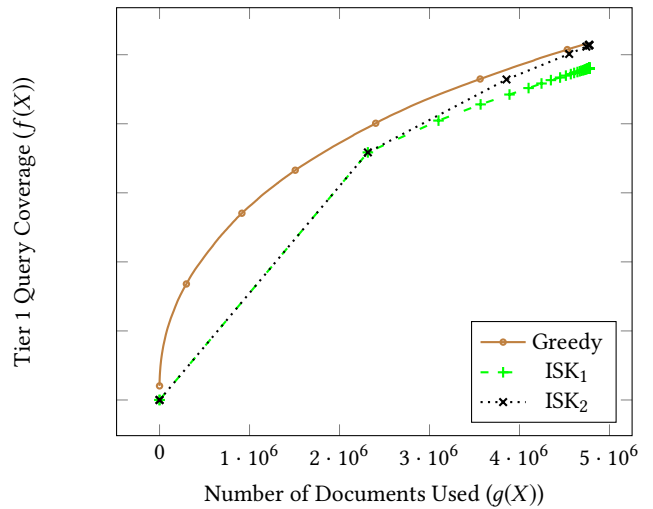
\begin{figure}
  \begin{tikzpicture}
    \begin{axis}[
      xlabel={Number of Documents Used ($g(X)$)},
      ylabel={Tier 1 Query Coverage ($f(X)$)},
      yticklabels={,,},
      legend pos=south east,
      scaled y ticks = false,
      scaled x ticks = false]

      \addplot+[color=brown, mark=o, mark repeat=100, mark size=1pt, thick] table [col sep=comma,
      x=con,y=obj] {optpesgreedy.txt}; \addlegendentry{Greedy}

      \addplot+[color=green, loosely dashed, mark=+, mark options={solid}, thick] table [col sep=comma,
      x=con,y=obj] {isk2.txt}; \addlegendentry{ISK$_1$}

      \addplot+[color=black, dotted, thick, mark=x, mark options={solid}, thick] table [col sep=comma,
      x=con,y=obj] {isk1.txt}; \addlegendentry{ISK$_2$}

    \end{axis}
  \end{tikzpicture}
  \caption{Solution path of different optimization algorithms.  Units
    in the $y$ axes are hidden for confidentiality. Intermediate
    solutions are marked on the line.}
  \label{fig:solutionpath}
\end{figure}

Also note that the greedy algorithm has an advantage of finding the
entire solution path for different values of the capacity parameter up
to $B$, as any intermediate solution $X^t$ can be considered as the
\emph{final} solution of the problem (\ref{eq:submod_obj}) with the
parameter $B$ set as $g(X^t)$. This is useful when the suitable size
of Tier 1 is not given a priori, and we need to search for the optimal
configuration of $B$. Figure~\ref{fig:solutionpath} shows the solution
path of different algorithms. ISK algorithms have very few
intermediate solutions which can be useful for determining $B$,
whereas the greedy algorithm has a very continuous solution path.

While Constraint-Agnostic Greedy algorithm is much faster than other
greedy algorithms because it ignores the constraint in the selection
process, it converges to a clearly suboptimal solution for the same
reason.  Therefore, the greedy procedure we propose
(\ref{eq:greedy_update}) seems to be more effective than
Constraint-Agnostic Greedy from \citet{iyer2013submodular}.  Among
algorithms which implement the proposed greedy procedure,
Opt./Pes. Greedy is the fastest because it leverages both
parallelization and lazy evaluation. This is confirmed in
Figure~\ref{fig:threads}; when smaller number of CPUs are used, the
efficiency gap between Lazy Greedy and Opt./Pes. Greedy becomes
narrower.

\begin{figure}
  \centering
  \begin{subfigure}[t]{0.25\textwidth}
    \centering
    \begin{tikzpicture}[scale=0.5]
      \begin{axis}[
        xlabel={Wall Clock Time (seconds)},
        ylabel={Tier 1 Query Coverage ($f(X)$)},
        yticklabels={,,},
        legend pos=south east,
        scaled y ticks = false]
        
        \addplot+[color=red, dashed, no markers, thick] table [col sep=comma,
        x=time,y=obj] {lazygreedy_thread2.txt}; \addlegendentry{Lazy Greedy}
        
        \addplot+[color=brown, densely dashed, no markers, thick] table [col sep=comma,
        x=time,y=obj] {optpesgreedy_thread2.txt}; \addlegendentry{Opt./Pes. Greedy}

        \addplot+[color=green, loosely dashed, no markers, thick] table [col sep=comma,
        x=time,y=obj] {mmshrink_thread2.txt}; \addlegendentry{ISK$_1$}

        \addplot+[color=black, dotted, thick, no markers, thick] table [col sep=comma,
        x=time,y=obj] {mmgrow_thread2.txt}; \addlegendentry{ISK$_2$}
      
      \end{axis}
    \end{tikzpicture}
    \caption{2 CPUs}
  \end{subfigure}
  ~
  \begin{subfigure}[t]{0.25\textwidth}
    \centering
    \begin{tikzpicture}[scale=0.5]
      \begin{axis}[
        xlabel={Wall Clock Time (seconds)},
        yticklabels={,,},
        legend pos=south east,
        scaled y ticks = false]
        
        \addplot+[color=red, dashed, no markers, thick] table [col sep=comma,
        x=time,y=obj] {lazygreedy_thread4.txt}; \addlegendentry{Lazy Greedy}
        
        \addplot+[color=brown, densely dashed, no markers, thick] table [col sep=comma,
        x=time,y=obj] {optpesgreedy_thread4.txt}; \addlegendentry{Opt./Pes. Greedy}

        \addplot+[color=green, loosely dashed, no markers, thick] table [col sep=comma,
        x=time,y=obj] {mmshrink_thread4.txt}; \addlegendentry{ISK$_1$}

        \addplot+[color=black, dotted, thick, no markers, thick] table [col sep=comma,
        x=time,y=obj] {mmgrow_thread4.txt}; \addlegendentry{ISK$_2$}
      
      \end{axis}
    \end{tikzpicture}
    \caption{4 CPUs}
  \end{subfigure}  
    
  \caption{Comparison of the performance of optimization algorithms in
    terms of the objective function $f(X)$ as a function of wall clock
    time elapsed, when different numbers of CPUs are used. $B$ was set
    as $\frac{\abr{\Dcal}}{4}$. Lazy Greedy was terminated before
    reaching the stopping condition.  Units in the $y$-axis are hidden
    for the confidentiality.}
  \label{fig:threads}
\end{figure}
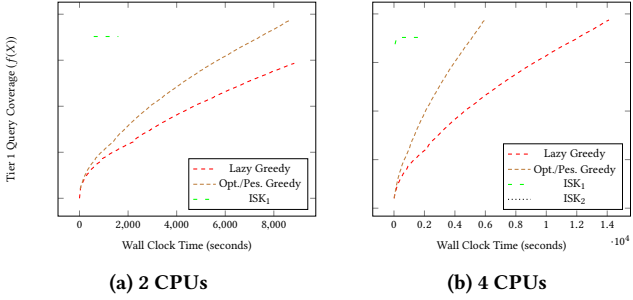

\subsection{Tiering Performance}

We compare the performance of tiering methods with respect to the
fraction of queries covered by Tier 1, which index is half the size of
the full index ($B = \frac{\abr{\Dcal}}{2}$). We consider following
methods:

\begin{description}
\item [\texttt{popularity}] An intuitive baseline from
  \citet{leung2010optimal} which selects $B$ most frequently appearing
  documents in the training data. That is, the score for each document
  $d$ is defined as
    $\PP_{q \sim \Qcal_n}\sbr{d \in m(q)}$,
  and top
  $B$ documents with respect to this score are chosen as
  $\Dcal_1$. Then, queries in $Q_n$ which match set is contained in
  these $B$ documents are set as $X^{\texttt{flow}}$.
\item [\texttt{flow-max}] The score for each document $d$ is defined
  as
  \begin{align*}
    \max_{q \in Q_n, d \in m(q)} \PP_{q \sim \Qcal_n}\sbr{q = q'},
  \end{align*}
  which is the maximum probability of the query which match set
  includes the document. Top $B$ documents according to this score are
  chosen, and again queries in $Q_n$ which match set is contained in
  these $B$ documents are set as $X^{\texttt{flow}}$. This rule is
  derived from the subgradient of the \texttt{flow}'s objective
  function. In \citet{leung2010optimal}, this simple heuristic was
  very competitive to principled optimization algorithms.
\item [\texttt{flow-sgd}] Convex relaxation of (\ref{eq:flowobj}),
  which is a maximum flow problem, is optimized with stochastic
  gradient descent. In order to provide regularization to this method,
  we introduced a hyperparameter $\lambda$, and removed any query
  which frequency is lower than the parameter; i.e., any $q \in Q_n$
  with $\PP_{q' \sim \Qcal_n}\sbr{q = q'} < \lambda$ are removed from
  the training data, so that the method does not overfit to rare
  queries.
\item [\texttt{clause}:] The clause selection-based query and document
  classifiers we proposed in Section~ \ref{ssec:clause_classifier}
  optimized with Opt./Pes. Greedy.
\end{description}

\begin{figure}
  \begin{tikzpicture}
    \begin{axis}[
      xlabel={Training Coverage},
      ylabel={Relative Test Coverage (\%)},
      ymin=0,
      xmin=40, xmax=70,
      xticklabels={,,},
      legend pos=south east]

      \addplot[color=brown, mark=*, only marks] coordinates {
        (44.1, 23.9)
      };
      \addlegendentry{\texttt{popularity}}

      \addplot[color=green, mark=+, only marks] coordinates {
        (45.1, 26.4)
      };
      \addlegendentry{\texttt{flow-max}}

      \addplot[color=red, mark=x] coordinates {
        (63.8,100.0)
        (53.9,54.2)
        (48.5,33.8)
        (47.1,29.0)
        (46.1,25.4)        
      };
      \addlegendentry{\texttt{flow-sgd}}

      \addplot[color=blue, mark=o] coordinates {
        (41.9, 159.9)
        (50.0, 186.5)
        (51.1, 188.6)
        (65.5, 200.6)        
      };
      \addlegendentry{\texttt{clause}}

      \addplot[mark=none, black, samples=2, domain=0:100] {100.0};
      
    \end{axis}
  \end{tikzpicture}
  \caption{Comparison of the Tier 1 query coverage (the fraction of
    queries classified into Tier 1) in training data and test
    data. $y$-axis shows the coverage on the test data relative to the
    best result of \texttt{flow-sgd}.  For \texttt{flow-sgd} and
    \texttt{clause}, each point corresponds to a solution with a
    different regularization parameter $\lambda$. Units in the $x$
    axis are hidden for the confidentiality.}
  \label{fig:general}
\end{figure}
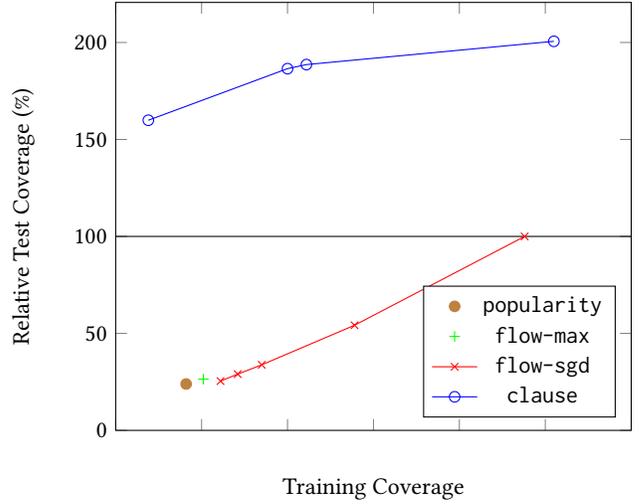

Figure~\ref{fig:general} simultaneously compares each method's fit to
the training data and their generalization to the future traffic.
Simple heuristics \texttt{popularity} and \texttt{flow-max} had a very
poor fit to the training data, although they were very competitive in
\citet{leung2010optimal}. This is probably because they only
considered top few documents per query, whereas we aim to cover the
entire match set $m(q)$; the larger the number of documents per query
$\abr{m(q)}$, the less likely it is for the heuristically chosen set
of documents $\Dcal_1$ to cover the entire match set.  Since
\texttt{flow-sgd} is a principled optimization algorithm, it is
successful at enforcing the correctness of query classification, and
achieves as high coverage in the training data as that of the
\texttt{clause} method we propose. However, the generalization
performance of \texttt{flow-sgd} was poorer than \texttt{clause}, and
exploring different values of the regularization parameter did not
help.  This demonstrates the effectiveness of our regularized
empirical risk minimization approach for optimizing the generalization
performance.

\section{Conclusion and Future Work}

We demonstrated that the clause selection-based tiering method is more
effective than the query selection-based method in terms of the
generalization performance to the future traffic. We proposed multiple
algorithms for optimizing the configuration of the tiered
architecture, which can also be used for general submodular
maximization problems with a submodular upper bound constraint. We
showed that these algorithms can scale to practical problems by
evaluating on real-world data.

In the future, it would be of interest to study how the proposed
method can be generalized for an arbitrary number of
tiers. Theoretical studies of the algorithms we proposed would be also
insightful. Indeed, it is quite surprising that very different
discrete optimization algorithms are converging to similar solutions;
there is probably a structure in the problem stronger than SCSK, which
helps algorithms to avoid being stuck in bad local minima.

\bibliographystyle{ACM-Reference-Format}
\bibliography{formbib}

\end{document}